\documentclass{article}
\usepackage{fullpage}
\usepackage{amsmath,amsfonts,amssymb,amsthm}
\usepackage{graphicx} 
\usepackage{hyperref}
\usepackage[nameinlink]{cleveref}
\usepackage[T1]{fontenc}
\usepackage{physics}
\usepackage{todonotes}
\allowdisplaybreaks

\hypersetup{colorlinks={true},linkcolor={blue},citecolor=magenta}

\usepackage{comment}

\usepackage[style=alphabetic,minalphanames=3,maxalphanames=4,maxnames=99,backref=true]{biblatex}
\title{A Note on the Common Haar State Model}

\author{Prabhanjan Ananth\thanks{\texttt{prabhanjan@cs.ucsb.edu}}\\ {\small UCSB} \and Aditya Gulati\thanks{\texttt{adityagulati@ucsb.edu}} \\ {\small UCSB} \and Yao-Ting Lin\thanks{\texttt{yao-ting\_lin@ucsb.edu}} \\ {\small UCSB}}
\date{}

\begin{document}

\maketitle

\newtheorem{theorem}{Theorem}[section]

\theoremstyle{definition}
\newtheorem{definition}[theorem]{Definition}

\newtheorem{lemma}[theorem]{Lemma}

\newtheorem{fact}[theorem]{Fact}

\newtheorem{proposition}[theorem]{Proposition}

\newtheorem{corollary}[theorem]{Corollary}

\newtheorem{remark}[theorem]{Remark}


\newcommand{\etal}{et~al.\ }
\newcommand{\aka}{also known as,\ }
\newcommand{\resp}{resp.,\ }
\newcommand{\ie}{i.e.,\ }
\newcommand{\wolog}{w.l.o.g.\ }
\newcommand{\Wolog}{W.l.o.g.\ }
\newcommand{\eg}{e.g.,\ }
\newcommand{\wrt} {with respect to\ }
\newcommand{\cf}{{cf.,\ }}

\newcommand{\st}{\ \text{s.t.}\ }


\newcommand{\N}{\mathbb{N}}
\newcommand{\Q}{\mathbb{Q}}
\newcommand{\R}{\mathbb{R}}
\newcommand{\Z}{\mathbb{Z}}

\newcommand{\round}[1]{\lfloor #1 \rceil}
\newcommand{\ceil}[1]{\lceil #1 \rceil}
\newcommand{\floor}[1]{\lfloor #1 \rfloor}
\newcommand{\angles}[1]{\langle #1 \rangle}
\newcommand{\parens}[1]{( #1 )}
\newcommand{\bracks}[1]{[ #1 ]}
\renewcommand{\bra}[1]{\langle#1\rvert}
\renewcommand{\braket}[2]{\langle #1 \mid #2 \rangle}
\renewcommand{\ket}[1]{\lvert#1\rangle}
\newcommand{\set}[1]{\{ #1 \}}
\newcommand{\bit}{\{0,1\}}


\newcommand{\cA}{{\mathcal A}}
\newcommand{\cB}{{\mathcal B}}
\newcommand{\cC}{{\mathcal C}}
\newcommand{\cD}{{\mathcal D}}
\newcommand{\cE}{{\mathcal E}}
\newcommand{\cF}{{\mathcal F}}
\newcommand{\cG}{{\mathcal G}}
\newcommand{\cH}{{\mathcal H}}
\newcommand{\cI}{{\mathcal I}}
\newcommand{\cJ}{{\mathcal J}}
\newcommand{\cK}{{\mathcal K}}
\newcommand{\cL}{{\mathcal L}}
\newcommand{\cM}{{\mathcal M}}
\newcommand{\cN}{{\mathcal N}}
\newcommand{\cO}{{\mathcal O}}
\newcommand{\cP}{{\mathcal P}}
\newcommand{\cQ}{{\mathcal Q}}
\newcommand{\cR}{{\mathcal R}}
\newcommand{\cS}{{\mathcal S}}
\newcommand{\cT}{{\mathcal T}}
\newcommand{\cU}{{\mathcal U}}
\newcommand{\cV}{{\mathcal V}}
\newcommand{\cW}{{\mathcal W}}
\newcommand{\cX}{{\mathcal X}}
\newcommand{\cY}{{\mathcal Y}}
\newcommand{\cZ}{{\mathcal Z}}

\newcommand{\bfA}{\mathbf{A}}
\newcommand{\bfB}{\mathbf{B}}
\newcommand{\bfC}{\mathbf{C}}
\newcommand{\bfD}{\mathbf{D}}
\newcommand{\bfE}{\mathbf{E}}
\newcommand{\bfF}{\mathbf{F}}
\newcommand{\bfG}{\mathbf{G}}
\newcommand{\bfH}{\mathbf{H}}
\newcommand{\bfI}{\mathbf{I}}
\newcommand{\bfJ}{\mathbf{J}}
\newcommand{\bfK}{\mathbf{K}}
\newcommand{\bfL}{\mathbf{L}}
\newcommand{\bfM}{\mathbf{M}}
\newcommand{\bfN}{\mathbf{N}}
\newcommand{\bfO}{\mathbf{O}}
\newcommand{\bfP}{\mathbf{P}}
\newcommand{\bfQ}{\mathbf{Q}}
\newcommand{\bfR}{\mathbf{R}}
\newcommand{\bfS}{\mathbf{S}}
\newcommand{\bfT}{\mathbf{T}}
\newcommand{\bfU}{\mathbf{U}}
\newcommand{\bfV}{\mathbf{V}}
\newcommand{\bfW}{\mathbf{W}}
\newcommand{\bfX}{\mathbf{X}}
\newcommand{\bfY}{\mathbf{Y}}
\newcommand{\bfZ}{\mathbf{Z}}

\newcommand{\bfa}{\mathbf{a}}
\newcommand{\bfb}{\mathbf{b}}
\newcommand{\bfc}{\mathbf{c}}
\newcommand{\bfd}{\mathbf{d}}
\newcommand{\bfe}{\mathbf{e}}
\newcommand{\bff}{\mathbf{f}}
\newcommand{\bfg}{\mathbf{g}}
\newcommand{\bfh}{\mathbf{h}}
\newcommand{\bfi}{\mathbf{i}}
\newcommand{\bfj}{\mathbf{j}}
\newcommand{\bfk}{\mathbf{k}}
\newcommand{\bfl}{\mathbf{l}}
\newcommand{\bfm}{\mathbf{m}}
\newcommand{\bfn}{\mathbf{n}}
\newcommand{\bfo}{\mathbf{o}}
\newcommand{\bfp}{\mathbf{p}}
\newcommand{\bfq}{\mathbf{q}}
\newcommand{\bfr}{\mathbf{r}}
\newcommand{\bfs}{\mathbf{s}}
\newcommand{\bft}{\mathbf{t}}
\newcommand{\bfu}{\mathbf{u}}
\newcommand{\bfv}{\mathbf{v}}
\newcommand{\bfw}{\mathbf{w}}
\newcommand{\bfx}{\mathbf{x}}
\newcommand{\bfy}{\mathbf{y}}
\newcommand{\bfz}{\mathbf{z}}

\newcommand{\sfA}{\mathsf{A}}
\newcommand{\sfB}{\mathsf{B}}
\newcommand{\sfC}{\mathsf{C}}
\newcommand{\sfD}{\mathsf{D}}
\newcommand{\sfE}{\mathsf{E}}
\newcommand{\sfF}{\mathsf{F}}
\newcommand{\sfG}{\mathsf{G}}
\newcommand{\sfH}{\mathsf{H}}
\newcommand{\sfI}{\mathsf{I}}
\newcommand{\sfJ}{\mathsf{J}}
\newcommand{\sfK}{\mathsf{K}}
\newcommand{\sfL}{\mathsf{L}}
\newcommand{\sfM}{\mathsf{M}}
\newcommand{\sfN}{\mathsf{N}}
\newcommand{\sfO}{\mathsf{O}}
\newcommand{\sfP}{\mathsf{P}}
\newcommand{\sfQ}{\mathsf{Q}}
\newcommand{\sfR}{\mathsf{R}}
\newcommand{\sfS}{\mathsf{S}}
\newcommand{\sfT}{\mathsf{T}}
\newcommand{\sfU}{\mathsf{U}}
\newcommand{\sfV}{\mathsf{V}}
\newcommand{\sfW}{\mathsf{W}}
\newcommand{\sfX}{\mathsf{X}}
\newcommand{\sfY}{\mathsf{Y}}
\newcommand{\sfZ}{\mathsf{Z}}

\newcommand{\sfa}{\mathsf{a}}
\newcommand{\sfb}{\mathsf{b}}
\newcommand{\sfc}{\mathsf{c}}
\newcommand{\sfd}{\mathsf{d}}
\newcommand{\sfe}{\mathsf{e}}
\newcommand{\sff}{\mathsf{f}}
\newcommand{\sfg}{\mathsf{g}}
\newcommand{\sfh}{\mathsf{h}}
\newcommand{\sfi}{\mathsf{i}}
\newcommand{\sfj}{\mathsf{j}}
\newcommand{\sfk}{\mathsf{k}}
\newcommand{\sfl}{\mathsf{l}}
\newcommand{\sfm}{\mathsf{m}}
\newcommand{\sfn}{\mathsf{n}}
\newcommand{\sfo}{\mathsf{o}}
\newcommand{\sfp}{\mathsf{p}}
\newcommand{\sfq}{\mathsf{q}}
\newcommand{\sfr}{\mathsf{r}}
\newcommand{\sfs}{\mathsf{s}}
\newcommand{\sft}{\mathsf{t}}
\newcommand{\sfu}{\mathsf{u}}
\newcommand{\sfv}{\mathsf{v}}
\newcommand{\sfw}{\mathsf{w}}
\newcommand{\sfx}{\mathsf{x}}
\newcommand{\sfy}{\mathsf{y}}
\newcommand{\sfz}{\mathsf{z}}
\newcommand{\eps}{\varepsilon}
\newcommand{\veps}{\varepsilon}
\newcommand{\vphi}{\varphi}
\newcommand{\vsigma}{\varsigma}
\newcommand{\vrho}{\varrho}
\newcommand{\vpi}{\varpi}


\newcommand{\secp}{{\lambda}}
\newcommand{\poly}{\operatorname{poly}}
\newcommand{\polylog}{\operatorname{polylog}}
\newcommand{\loglog}{\operatorname{loglog}}
\newcommand{\GF}{\operatorname{GF}}
\newcommand{\Exp}{\operatorname*{\mathbb{E}}}
\newcommand{\Ex}{\Exp}
\newcommand{\View}[1]{\mathsf{View}\angles{#1}}
\newcommand{\Var}{\operatorname*{Var}}
\newcommand{\ShH}{\operatorname{H}}
\newcommand{\maxH}{\operatorname{H_0}}
\newcommand{\minH}{\operatorname{H_{\infty}}}
\newcommand{\Enc}{\operatorname{Enc}}
\newcommand{\Setup}{\operatorname{Setup}}
\newcommand{\KGen}{\operatorname{KeyGen}}
\newcommand{\StGen}{\operatorname{StateGen}}
\newcommand{\Dec}{\operatorname{Dec}}
\newcommand{\Sign}{\operatorname{Sign}}
\newcommand{\Ver}{\operatorname{Ver}}
\newcommand{\Gen}{\operatorname{Gen}}
\newcommand{\negl}{\operatorname{negl}}
\newcommand{\Supp}{\operatorname{Supp}}
\newcommand{\maj}{\operatorname*{maj}}
\newcommand{\argmax}{\operatorname*{arg\,max}}
\newcommand{\Image}{\operatorname{Im}}

\newcommand{\TD}{\mathsf{TD}}


\newcommand{\class}[1]{\ensuremath{\mathbf{#1}}}
\newcommand{\coclass}[1]{\class{co\mbox{-}#1}} 
\newcommand{\prclass}[1]{\class{Pr#1}}
\newcommand{\PPT}{\class{PPT}}
\newcommand{\BPP}{\class{BPP}}
\newcommand{\NC}{\class{NC}}
\newcommand{\AC}{\class{AC}}
\newcommand{\NP}{\ensuremath{\class{NP}}}
\newcommand{\coNP}{\ensuremath{\coclass{NP}}}
\newcommand{\RP}{\class{RP}}
\newcommand{\coRP}{\coclass{RP}}
\newcommand{\ZPP}{\class{ZPP}}
\newcommand{\IP}{\class{IP}}
\newcommand{\coIP}{\coclass{IP}}
\newcommand{\AM}{\class{AM}}
\newcommand{\coAM}{\class{coAM}}
\newcommand{\MA}{\class{MA}}
\renewcommand{\P}{\class{P}}
\newcommand\prBPP{\prclass{BPP}}
\newcommand\prRP{\prclass{RP}}
\newcommand{\Ppoly}{\class{P/poly}}
\newcommand{\DTIME}{\class{DTIME}}
\newcommand{\ETIME}{\class{E}}
\newcommand{\BPTIME}{\class{BPTIME}}
\newcommand{\AMTIME}{\class{AMTIME}}
\newcommand{\coAMTIME}{\class{coAMTIME}}
\newcommand{\NTIME}{\class{NTIME}}
\newcommand{\coNTIME}{\class{coNTIME}}
\newcommand{\AMcoAM}{\class{AM} \cap \class{coAM}}
\newcommand{\NPcoNP}{\class{N} \cap \class{coN}}
\newcommand{\EXP}{\class{EXP}}
\newcommand{\SUBEXP}{\class{SUBEXP}}
\newcommand{\qP}{\class{\wt{P}}}
\newcommand{\PH}{\class{PH}}
\newcommand{\NEXP}{\class{NEXP}}
\newcommand{\PSPACE}{\class{PSPACE}}
\newcommand{\NE}{\class{NE}}
\newcommand{\coNE}{\class{coNE}}
\newcommand{\SharpP}{\class{\#P}}
\newcommand{\prSZK}{\prclass{SZK}}
\newcommand{\BPPSZK}{\BPP^{\SZK}}
\newcommand{\TFAM}{\class{TFAM}}
\newcommand{\RTFAM}{{\R\text{-}\class{TUAM}}}
\newcommand{\FBPP}{\class{BPP}}
\newcommand{\BPPR}{\FBPP^{\RTFAM}}
\newcommand{\TFNP}{\mathbf{TFNP}}
\newcommand{\MIP}{\mathbf{MIP}}
\newcommand{\ZK}{\class{ZK}}
\newcommand{\SZK}{\class{SZK}}
\newcommand{\SZKP}{\class{SZKP}}
\newcommand{\SZKA}{\class{SZKA}}
\newcommand{\CZKP}{\class{CZKP}}
\newcommand{\CZKA}{\class{CZKA}}
\newcommand{\coSZKP}{\coclass{SZKP}}
\newcommand{\coSZKA}{\coclass{SZKA}}
\newcommand{\coCZKP}{\coclass{CZKP}}
\newcommand{\coCZKA}{\coclass{CZKA}}
\newcommand{\NISZKP}{\class{NI\mbox{-}SZKP}}
\newcommand{\phcc}{\ensuremath \PH^{\mathrm{cc}}}
\newcommand{\pspacecc}{\ensuremath \PSPACE^{\mathrm{cc}}}
\newcommand{\sigcc}[1]{\ensuremath \class{\Sigma}^{\mathrm{cc}}_{#1}}
\newcommand{\aco}{\ensuremath \AC^0}
\newcommand{\ncone}{\ensuremath \NC^1}

\newcommand{\shadowgen}{\mathsf{ShadowGen}}
\newcommand{\E}{\mathop{\mathbb{E}}}
\newcommand{\commit}{{\mathsf{commit}}}
\newcommand{\reveal}{{\mathsf{reveal}}}
\newcommand{\aux}{{\mathsf{aux}}}
\newcommand{\com}{{\mathsf{com}}}
\newcommand{\swapt}{{\mathsf{SwapTest}}}
\newcommand{\accept}{{\mathsf{accept}}}
\newcommand{\C}{\mathbb{C}}
\newcommand{\Haar}{\mathcal{H}}
\newcommand{\sym}{\mathsf{Sym}}
\renewcommand{\ketbra}[2]{\ket{#1}\bra{#2}}
\renewcommand{\braket}[2]{\langle #1 | #2 \rangle}
\newcommand{\inner}[2]{\langle #1, #2 \rangle}
\newcommand{\ugets}{\xleftarrow{\$}}
\newcommand{\ol}[1]{\overline{#1}}

\newcommand{\gap}{k} 
\newcommand{\keylength}{m} 

\newcommand{\supp}{\mathsf{supp}}
\newcommand{\freq}[2]{\mathsf{freq}_{#2} (#1)}
\newcommand{\hamming}{\mathsf{size}}
\newcommand{\type}{\mathsf{type}}
\newcommand{\bintype}{\mathsf{bintype}}
\newcommand{\setT}{\mathsf{mset}}
\newcommand{\rep}{\mathsf{rep}}
\newcommand{\goodpre}[4]{\mathcal{I}^{(#3)}_{#1,#2} (#4)}

\newcommand{\good}{\mathsf{Good}}
\newcommand{\bad}{\mathsf{Bad}}
\begin{abstract}
\noindent Common random string model is a popular model in classical cryptography with many constructions proposed in this model. We study a quantum analogue of this model called the common Haar state model, which was also studied in an independent work by Chen, Coladangelo and Sattath (arXiv 2024). In this model, every party in the cryptographic system receives many copies of one or more i.i.d Haar states. 
\par Our main result is the construction of a statistically secure PRSG with: (a) the output length of the PRSG is strictly larger than the key size, (b) the security holds even if the adversary receives $O\left(\frac{\lambda}{(\log(\lambda))^{1.01}} \right)$ copies of the pseudorandom state. We show the optimality of our construction by showing a matching lower bound. Our construction is simple and its analysis uses elementary techniques.
\end{abstract}
\tableofcontents 
\section{Introduction} 
In classical cryptography, the common random string and the common reference string models were introduced to tackle cryptographic tasks that were impossible to achieve in the plain model. In the common reference string model, there is a trusted setup who produces a string that every party has access to. In the common random string model, the common string available to all the parties is sampled uniformly at random, thus avoiding the need for a trusted setup. As a result, the common random string model is the more desirable model of the two. There have been many constructions proposed over the years in these two models, including non-interactive zero-knowledge~\cite{BFM19}, secure computation secure under universal composition~\cite{CF01,CLOS02} and two-round secure computation~\cite{GS19,BL19}. 
\par It is a worthy pursuit to study similar models for quantum cryptographic protocols. In this case, there is an option to define models that are inherently quantum. For instance, we could define a model wherein a trusted setup produces a quantum state and every party participating in the cryptographic system receives one or many copies of this quantum state. Indeed, two works by Morimae, Nehoran and Yamakawa~\cite{MNY23} and Qian~\cite{Qian23} consider this model, termed as the {\em common quantum reference string model} (CQRS). They proposed unconditionally secure commitments in this model.  Quantum commitments is a foundational notion in quantum cryptography. In recent years, quantum commitments have been extensively studied~\cite{AQY21,MY21,AGQY22,MY23onewayness,BCQ23,Bra23} due to its implication to secure computation~\cite{BartusekCKM21a,GLSV21}. The fact that information-theoretically secure commitments are impossible in the plain model~\cite{LoChau97,May97,CLM23} makes the contributions of~\cite{MNY23,Qian23} quite interesting. 
\par While CQRS is a quantum analogue of common reference string model, we can ask if there is a quantum analogue of the common random string model. An independent and concurrent recent work by Chen, Coladangelo and Sattath~\cite{chen2024power} (henceforth, referred to as CCS) introduced a model, called the {\em common Haar random state model} (CHRS). In this model, every party in the system receives many copies of many i.i.d Haar states. They presented constructions of pseudorandomness and commitments in this model. The goal of our work is to study this model further. 

\paragraph{Our Work.} We consider a simpler version of the common Haar random state model, wherein every party receives many copies of one Haar state. Feasibility results in our model are stronger than the model considered by CCS whereas the negative results would be weaker. We call our model the common Haar state model (CHS).  
\par Similar to CCS, we present constructions of pseudorandom states and commitments in the CHS model. \\

\noindent \underline{\textsc{Statistical Pseudorandom States (i.e., State Designs)}}: The concept of pseudorandom state generators (PRSGs) was introduced in a seminal work by Ji, Liu and Song~\cite{JLS18}. Roughly speaking, it states that any computationally bounded adversary cannot distinguish whether it receives many copies of a state produced using a pseudorandom state generator on a uniform key versus many copies of a single Haar state. Constructions of pseudorandom states are known from one-way functions~\cite{JLS18,BS19,AGQY22,JGMW23,giurgica2023pseudorandomness}. It is known that for a broad range of parameters, we need to impose the restriction that the adversary is computationally bounded~\cite{AGQY22}. One such case is when the output length of the state generator is larger than the length of the key; referred to as {\em stretch} PRSGs. In this case, even if the adversary gets one copy of the state then it can break the pseudorandomness property if it is unbounded. 
\par We consider the possibility of statistically secure stretch PRSGs. We show the following. In the theorem below, we denote $\lambda$ to be the key length of the PRSG. 

\newcommand{\secparam}{\lambda}
\begin{theorem}[Informal]
\label{thm:intro:prsgs}
There is a statistically secure PRSG in the CHS model satisfying the following: (a) the output length of PRSG is $> \lambda$ and, (b) the security holds as long as the adversary receives $O\left( \frac{\secparam}{\left( \log(\secparam) \right)^{1+\varepsilon} }\right)$ copies, for some constant $\varepsilon > 0$. 
\end{theorem} 

\noindent This improves upon CCS~\cite{chen2024power} who only showed 1-copy stretch PRSG exists in their model. Moreover, our construction and in particular, our analysis, is arguably simpler and elementary compared to CCS. Another advantage of our work is that we can achieve arbitrary stretch whereas it is unclear whether this is also achieved by CCS. As a side contribution, the proof of our PRSG construction also simplifies the proof of the quantum public-key construction of Coladangelo~\cite{Col23}; this is due to the fact the core lemma proven in~\cite{Col23} is a paraphrased version of the above theorem. 
\par Interestingly, both works design PRSGs which only consume one copy of a single Haar state. In this special case, it is interesting to understand whether we can extend our result to the setting when the adversary receives $O\left( \frac{\secparam}{\log(\secparam)}\right)$ copies. We show this is not possible. 

\begin{theorem}[Informal]
 There does not exist a secure PRSG in the CHS model satisfying the following: (a) the PRSG uses only one copy of the Haar state in the CHS model, (b) the adversary receives $\Omega(\frac{\secparam}{\log(\secparam)})$ copies of the pseudorandom state and, (c) the output length is $\omega(\log(\secparam))$.      
\end{theorem}

\noindent CCS also proved a lower bound where they showed that unbounded copy pseudorandom states do not exist. Their negative result is stronger in the sense that they rule out PRSGs who use up many copies of the Haar states from the CHRS. On the other hand, for the special case when the PRSG only takes one copy of the Haar state, we believe our result yields better parameters. \\

\noindent \underline{\textsc{Commitments}}: We show the following. 

\begin{theorem}[Informal]
There is an unconditionally secure bit commitment scheme in the CHS model.  
\end{theorem}

\noindent Both our construction and the commitments scheme proposed by CCS are different although they share strong similarities. Even the proof techniques seem to be similar. 

\paragraph{Future Directions.} The work by~\cite{chen2024power} and our work initiates an interesting research direction: to build cryptography in the common Haar state model. It would be interesting to the feasibility of other cryptographic constructions in this model. It would also be interesting to understand other variants of this model. One such variant is the Haar random oracle model, studied by~\cite{BFV19,CM21}. We leave open the question of investigating the relationship between these different models and also developing toolkits that will aid us in proving (in)feasibility results in these models.

\section{Preliminaries}

We denote the security parameter by $\secp$. We assume that the reader is familiar with the fundamentals of quantum computing covered in~\cite{nielsen_chuang_2010}.

\subsection{Notation}
\begin{itemize}
\item We use $[n]$ to denote $\{1,\ldots,n\}$ and $[0:n]$ to denote $\{0,1,\ldots,n\}$. 
\item We denote $S_t$ to be the symmetric group of degree $t$. 
\item For a set $A$ and $t \in \N$, we define $A^t := \{(a_1,\ldots,a_t):\ \forall i\in[t],\  a_i \in A\}$. 

\item For $\sigma\in S_t$ and $\bfv = (v_1,\ldots,v_t)$, we define $\sigma(\bfv) := (v_{\sigma(1)},\ldots,v_{\sigma(t)})$.
\item We use ${\cal D}(H)$ to denote the set of density matrices in the Hilbert space $H$. 
\item Let $\rho_{AB} \in \cD(H_A\otimes H_B)$, by $\Tr_{B}(\rho_{AB}) \in \cD(H_A)$ we denote the reduced density matrix by taking partial trace over $B$.
\item We denote the trace distance between quantum states $\rho, \rho'$ by $\TD(\rho, \rho') := \frac{1}{2}\| \rho - \rho' \|_1$.
\item We denote the Haar measure over $n$ qubits by $\Haar_n$.
\end{itemize}

\subsection{Pseudorandom State Generators}
We recall the definition of statistical pseudorandom state generators (PRS). 
\begin{definition}
    We say that a QPT algorithm $G$ is an $\ell$-copy statistical PRS generator if the following holds:
    \begin{itemize}
        \item \textbf{State Generation}:
        For any $\secp\in\N$ and $k\in\set{0,1}^{\secp}$, the algorithm $G$ is a quantum channel that satisfies 
        \[
        G(\ket{k}) = \ketbra{\psi_k}{\psi_k},
        \]
        for some $n(\secp)$-qubit state $\ket{\psi_k}$.
        
        \item \textbf{Psuedorandomness}: For any computationally unbounded adversary $A$, there exists a negligible fuctions $\negl(\cdot)$ such that: 
        \[
        \left|
        \Pr_{\substack{k\gets\set{0,1}^{\secp}}}
        \left[ A_{\secp}\left(G(\ket{k})^{\otimes \ell}\right) = 1 \right]
        - \Pr_{\substack{\ket{\varphi}\gets\Haar_n}}
        \left[ A_{\secp}\left(\ket{\varphi}^{\otimes \ell}\right) = 1 \right]
        \right|
        \leq \negl(\secp).
        \]
    \end{itemize}
\end{definition}

\noindent If $G$ statisfies the above security definition for every polynomial $\ell$, we say that $G$ is an unbounded poly-copy statistical PRS generator.

\subsection{Quantum Commitments}
We recall the definition of commitment schemes in the CRQS model~\cite{MNY23}.
\begin{definition}[Quantum commitments in the Common Reference Quantum State (CRQS) model~\cite{MNY23}]
A (non-interactive) quantum commitment scheme in the CRQS model is given by a tuple of the setup algorithm $\Setup$, committer $C$, and receiver $R$, all of which are uniform QPT algorithms. The scheme is divided into three phases, the setup phase, commit phase, and reveal phase as follows:
    \begin{itemize}
        \item Setup phase: $\Setup$ takes $1^{\secp}$ as input, uniformly samples a classical key $k \gets K_{\secp}$, generates two copies of the same pure state $\ket{\psi_{k}}$ and sends one copy each to $C$ and $R$.
        \item Commit phase: $C$ takes $\ket{\psi_k}$ given by the setup algorithm and a bit $b \in \set{0, 1}$ to commit as input, generates a quantum state on registers $\sfC$ and $\sfR$, and sends the register $\sfC$ to $R$.
        \item Reveal phase: $C$ sends $b$ and the register $\sfR$ to $R$. $R$ takes $\ket{\psi_{k}}$ given by the setup algorithm and $(b, \sfC, \sfR)$ given by $C$ as input, and outputs $b$ if it accepts and otherwise outputs $\bot$.
    \end{itemize}
\end{definition}

\begin{definition}[$t$-copy statistical hiding~\cite{MNY23}]
A quantum commitment scheme $(\Setup, C, R)$ in the CRQS model satisfies $t$-copy statistical hiding if for any non-uniform unbounded-time algorithm $\cA$,
\begin{multline*}
\Bigg\vert \Pr[\cA(1^{\secp},\ket{\psi_k}^{\otimes t},\Tr_{\sfR}(\sigma_{\sfC\sfR})) = 1: \substack{k\gets K_{\secp},\\ \sigma_{\sfC\sfR}\gets C_{\com}(\ket{\psi_k},0)}] \\
- \Pr[\cA(1^{\secp},\ket{\psi_k}^{\otimes t},\Tr_{\sfR}(\sigma_{\sfC\sfR}))  = 1 :
\substack{k\gets K_{\secp}, \\ \sigma_{\sfC\sfR}\gets C_{\com}(\ket{\psi_k},1)}] \Bigg\vert
\leq \negl(\secp),
\end{multline*}
where $C_{\com}$ is the commit phase of $C$.
\end{definition}

\begin{definition}[Statistical sum-binding~\cite{MNY23}] A quantum commitment scheme $(\Setup, C, R)$ in the CQRS model satisfies statistical sum-binding if the following holds. For any pair of non-uniform unbounded-time malicious committers $C_0^{\ast}$ and $C_1^{\ast}$ that take the classical key $k$, which is sampled by the setup algorithm, as input and work in the same way in the commit phase, if we let $p_b$ to be the probability that $R$ accepts the revealed bit $b$ in the interaction with $C_b^{\ast}$ for $b \in \set{0, 1}$, then we have
$$p_0+p_1\leq 1+\negl(\secp).$$
\end{definition}

\subsection{Common Haar State Model}

The Common Haar State (CHS) model is a variant of the Common Reference Quantum State (CRQS) model. In this model, all parties receive polynomially many copies of a qubit state sampled from the Haar distribution. 
 
\subsubsection{Pseudorandom State Generators in the CHS model}

\begin{definition}[$\ell$-copy PRS in CHS model]
    Let $\ket{\vartheta}$ denote the $n(\secp)$-qubit common Haar state. We say that a QPT algorithm $G$ is an $\ell$-copy statistical PRS generator in the CHS model if the following holds:
    \begin{itemize}
        \item \textbf{State Generation}:
        For any $\secp\in\N$ and $k\in\set{0,1}^{\secp}$, the algorithm $G_k$ (where $G_k$ denotes $G(k,\cdot)$) is a quantum channel such that for every $n$-qubit state $\ket{\vartheta}$,
        \[
        G_k(\ketbra{\vartheta}{\vartheta}) = \ketbra{\vartheta_k}{\vartheta_k},
        \]
        for some $n$-qubit state $\ket{\vartheta_k}$. We sometimes write $G_k(\ket{\vartheta})$ for brevity.\footnote{More generally, the generation algorithm could take multiple copies of the common Haar state as input or output a state of size different from the CHS but we focus on generation algorithms taking only one copy of the Haar state and the output of the generator is the same size as the CHS.}
        \item \textbf{Pseudorandomness}: For any polynomial $t(\cdot)$ and computationally unbounded adversary $A$, there exists a negligible function $\negl(\cdot)$ such that: 
        \[
        \left|
        \Pr_{\substack{k\gets\set{0,1}^{\secp}\\ \ket{\vartheta}\gets\Haar_{n}}}
        \left[ A_{\secp}\left(G_k(\ket{\vartheta})^{\otimes \ell}, \ket{\vartheta}^{\otimes t}\right) = 1 \right]
        - \Pr_{\substack{\ket{\varphi}\gets\Haar_n\\ \ket{\vartheta}\gets\Haar_{n}}}
        \left[ A_{\secp}\left(\ket{\varphi}^{\otimes \ell}, \ket{\vartheta}^{\otimes t}\right) = 1 \right]
        \right|
        \leq \negl(\secp).
        \]
    \end{itemize}
\end{definition}
\noindent We define a stronger variant of the above notion called a multi-key $\ell$-copy PRS generator. Looking ahead, our construction of PRS in~\Cref{sec:prs-con} satisfies this stronger definition. In addition, we show in~\Cref{sec:Commitment} that multi-key $1$-copy stretch PRS generator in the CHS model implies statistically hiding, statistically sum-binding commitments in the CHS model.

\begin{definition}[Multi-key $\ell$-copy PRS in CHS model]
    Let $\ket{\vartheta}$ denote the $n(\secp)$-qubit common Haar state. We say that a QPT algorithm $G$ is a multi-key $\ell$-copy statistical PRS generator in the CHS model if the following holds:
    \begin{itemize}
        \item \textbf{State Generation}:
        For any $\secp\in\N$ and $k\in\set{0,1}^{\secp}$, the algorithm $G_k$ (where $G_k$ denotes $G(k,\cdot)$) is a quantum channel such that for every $n$-qubit state $\ket{\vartheta}$,
        \[
        G_k(\ketbra{\vartheta}{\vartheta}) = \ketbra{\vartheta_k}{\vartheta_k},
        \]
        for some $n$-qubit state $\ket{\vartheta_k}$. We sometimes write $G_k(\ket{\vartheta})$ for brevity.
        \item \textbf{Multi-key Pseudorandomness}: For any polynomial $t(\cdot)$, $p(\cdot)$ and computationally unbounded adversary $A$, there exists a negligible function $\negl(\cdot)$ such that: 
        \[
        \left|
        \Pr_{\substack{k_1,\ldots,k_p\gets\set{0,1}^{\secp}\\ \ket{\vartheta}\gets\Haar_{n}}}
        \left[ A_{\secp}\left(\otimes_{i=1}^p G_{k_i}(\ket{\vartheta})^{\otimes \ell}, \ket{\vartheta}^{\otimes t}\right) = 1 \right]
        - \Pr_{\substack{\ket{\varphi_1},\ldots,\ket{\varphi_p}\gets\Haar_n\\ \ket{\vartheta}\gets\Haar_{n}}}
        \left[ A_{\secp}\left(\otimes_{i=1}^{p}\ket{\varphi_i}^{\otimes \ell}, \ket{\vartheta}^{\otimes t}\right) = 1 \right]
        \right|
        \leq \negl(\secp).
        \]
    \end{itemize}
\end{definition}

\begin{remark}
    Note that in the plain model, PRS implies multi-key PRS because the PRS does not share randomness for different keys. This is not trivially true in CHS model as the generator for all keys shares the same common Haar state.
\end{remark}

\subsubsection{Quantum Commitments in the CHS model}
\begin{definition}[Quantum commitments in the Common Haar State (CHS) model]
A (non-interactive) quantum commitment scheme in the CHS model is given by a tuple of the committer $C$ and receiver $R$, all of which are uniform QPT algorithms. Let $\ket{\vartheta}$ be the $n(\secp)$-qubit common Haar state. The scheme is divided into two phases, commit phase, and reveal phase as follows:
    \begin{itemize}
        \item Commit phase: $C$ takes $\ket{\vartheta}^{\otimes p}$ for some polynomial $p(\cdot)$ and a bit $b \in \set{0, 1}$ to commit as input, generates a quantum state on registers $\sfC$ and $\sfR$, and sends the register $\sfC$ to $R$.
        \item Reveal phase: $C$ sends $b$ and the register $\sfR$ to $R$. $R$ takes $\ket{\vartheta}^{\otimes p}$ and $(b, \sfC, \sfR)$ given by $C$ as input, and outputs $b$ if it accepts and otherwise outputs $\bot$.
    \end{itemize}
\end{definition}

\begin{definition}[Poly-copy statistical hiding]
A quantum commitment scheme $(C, R)$ in the CHS model satisfies poly-copy statistical hiding if for any non-uniform unbounded-time algorithm $\cA$, and any polynomial $t(\cdot)$, there exists a negligible function $\negl(\cdot)$ such that
\begin{multline*}
\Bigg\vert \Pr[\cA(1^{\secp},\ket{\vartheta}^{\otimes t},\Tr_{\sfR}(\sigma_{\sfC\sfR})) = 1:
\substack{\ket{\vartheta}\gets \Haar_{n},\\ \sigma_{\sfC\sfR}\gets C_{\com}(\ket{\vartheta}^{\otimes p},0)}] \\
- \Pr[\cA(1^{\secp},\ket{\vartheta}^{\otimes t},\Tr_{\sfR}(\sigma_{\sfC\sfR})) = 1
:\substack{\ket{\vartheta}\gets \Haar_{n}, \\ \sigma_{\sfC\sfR}\gets C_{\com}(\ket{\vartheta}^{\otimes p},1)}] \Bigg\vert
\leq \negl(\secp),
\end{multline*}
where $C_{\com}$ is the commit phase of $C$.
\end{definition}

\begin{definition}[Statistical sum-binding] A quantum commitment scheme $(C, R)$ in the CHS model satisfies statistical sum-binding if the following holds. For any pair of non-uniform unbounded-time malicious committers $C_0^{\ast}$ and $C_1^{\ast}$ that take $\ket{\vartheta}^{\otimes T}$ for arbitrary large $T(\cdot)$ as input and work in the same way in the commit phase, if we let $p_b$ to be the probability that $R$ accepts the revealed bit $b$ in the interaction with $C_b^{\ast}$ for $b \in \set{0, 1}$, then we have
$$p_0+p_1\leq 1+\negl(\secp).$$
\end{definition}

\subsection{Symmetric Subspaces, Type States, and  Haar States}

\noindent The proof of facts and lemmas in this subsection can be found in~\cite{Harrow13church}. Let $\bfv = (v_1,\ldots,v_t)\in A^t$ for some finite set $A$. Let $|A| = N$. Define $\type(\bfv)\in [0:t]^N$ to be the \emph{type vector} such that the $i^{th}$ entry of $\type(\bfv)$ equals the number of occurrences of $i\in[N]$ in $\bfv$.\footnote{We identify $[0:t]^N$ as $[0:t]^A$.} In this work, by $T\in [0:t]^N$ we implicitly assume that $\sum_{i\in[N]}T_i = t$. 
For $T\in[0:t]^N$, we denote by $\setT(T)$ the \emph{multiset} uniquely determined by $T$. That is, the multiplicity of $i\in\setT(T)$ equals $T_i$ for all $i\in[N]$.
We write $T\gets [0:t]^N$ to mean sampling $T$ uniformly from $[0:t]^N$ conditioned on $\sum_{i\in[N]}T_i = t$. We write $\bfv\in T$ to mean $\bfv \in A^t$ satisfies $\type(\bfv) = T$.

In this work, we will focus on \emph{collision-free} types $T$ which satisfy $T_i\in\bit$ for all $i\in[N]$. A collision-free type $T$ can be naturally treated as a \emph{set} and we write $\bfv\gets T$ to mean sampling a uniform $\bfv$ conditioned on $\type(\bfv) = T$.

\begin{definition}[Type states] \label{def:type_states}
Let $T\in[0:t]^N$, we define the \emph{type states}:
\[
\ket{T} := \sqrt{\frac{\prod_{i\in[N]} T_i !}{t!}} \sum_{\bfv \in T} \ket{\bfv}.
\]
If $T$ is collision-free, then it can be simplified to 
\[
\ket{T} = \frac{1}{\sqrt{t!}} \sum_{\bfv \in T} \ket{\bfv}.
\]
Furthermore, it has the following useful expression
\begin{equation} \label{eq:useful}
\ketbra{T}{T} 
= \frac{1}{t!} \sum_{\bfv, \bfu \in T} \ketbra{\bfv}{\bfu}
= \Ex_{\bfv\gets T}\left[ \sum_{\sigma\in S_t}\ketbra{\bfv}{\sigma(\bfv)} \right].
\end{equation}
\end{definition}

\begin{lemma}[Average of copies of Haar-random states]
\label{fact:avg-haar-random}
For all $N,t \in \N$, we have
\[
\E_{\ket{\vartheta} \leftarrow \Haar(\C^N)} \ketbra{\vartheta}{\vartheta}^{\otimes t}
= \E_{T \leftarrow [0:t]^N} \ketbra{T}{T}.
\]
\end{lemma}

\begin{definition}[$\ell$-fold $n$-prefix collision-free types]
\label{def:l_fold_collisions}
Let $n,m,t,\ell\in\N$ such that $t \geq \ell$ and $T \in [0:t]^{2^{n+m}}$. We say $T$ is \emph{$\ell$-fold $n$-prefix collision-free} if for all pairs of $\ell$-subsets $\cS,\cT \subseteq \setT(T)$, the first $n$ bits of $\bigoplus_{x\in\cS} x \in \bit^{n+m}$ is identical to that of $\bigoplus_{y\in\cT} y \in \bit^{n+m}$ if and only if $\cS = \cT$. We define $\goodpre{n}{m}{\ell}{t} := \set{ T \in [0:t]^{2^{n+m}}: T \text{ is $\ell$-fold $n$-prefix collision-free} }$.
\end{definition}

\noindent For a fixed $t$, one can easily verify that $\ell$-fold $n$-prefix collision-freeness implies $\ell'$-fold $n$-prefix collision-freeness for $\ell>\ell'$, and $1$-fold $n$-prefix collision-freeness is equivalent to the standard collision-freeness.

\begin{lemma} \label{fact:random_type_l_fold_collision}
If $\ell = o(2^n)$, then $\Pr_{T\gets [0:t]^{2^{n+m}}}[ T \in \goodpre{n}{m}{\ell}{t}] 
= 1- O(t^{2\ell}/2^n)$.
\end{lemma}

\begin{proof}
First, sampling $T\gets [0:t]^{2^{n+m}}$ uniformly is $O(t^2/2^{n+m})$-close to sampling a uniform collision-free $T$ from $[0:t]^{2^{n+m}}$ by collision bound. Furthermore, sampling a uniform collision-free $T$ from $[0:t]^{2^{n+m}}$ is equivalent to sampling $t$ elements $x_1,x_2,\dots,x_t$ one by one from $\bit^{n+m}$ conditioned on them being distinct and setting $T$ such that $\setT(T) = \set{x_1,\ldots,x_t}$. Hence, it suffices to show that sampling $t$ elements $x_1,x_2,\dots,x_t$ one by one from $\bit^{n+m}$ conditioned on them being distinct results in an $\ell$-fold $n$-prefix collision-free set with probability $1- O(t^{2\ell}/2^n)$.

For any two distinct $\ell$-subsets of indices $\cS\neq\cT \subseteq [t]$, let $\bad_{\cS,\cT}$ denote the event that the first $n$ bits of $\bigoplus_{i\in\cS} x_i$ is the same as that of $\bigoplus_{j\in\cT}x_j$. Then $\Pr[\bad_{\cS,\cT}: \substack{x_1,x_2,\dots,x_t\gets \bit^{n+m}\\ x_1,x_2,\dots,x_t\text{ are distinct}}] = O(1/2^n-2\ell)$. This is because we can first sample $|\cS\cup\cT|-1$ elements (in $\cS\cup\cT$) except one with indices in $\cS\setminus\cT$. Then $\bad_{\cS,\cT}$ occurs only if the first $n$ bits of the last sample is equal to the first $n$ bits of the bitwise XOR of all other elements in $\cS$ with all elements in $\cT$, which happens with probability at most $O(1/(2^n-2\ell))$. By a union bound, we have $T \in \goodpre{n}{m}{\ell}{t}$ with probability one but $O(t^2/2^{n+m}) + \binom{t}{\ell}^2\cdot O(1/(2^n-2\ell)) = O(t^{2\ell}/2^n)$.
\end{proof}

\begin{lemma} \label{lem:perm_split}
For any $\bfv\in\bit^{(n+m)t}$ such that $\type(\bfv) \in \goodpre{n}{m}{\ell}{t}$ and $\sigma\in S_t$, define 
\[
A_{\bfv,\sigma} 
:= \E_{k\in\bit^{n}}\left[ \left(\left(Z^{k}\otimes I_m\right)^{\otimes \ell}\otimes I_{n+m}^{\otimes t-\ell}\right)\ketbra{\bfv}{\sigma(\bfv)}\left(\left(Z^{k}\otimes I_m\right)^{\otimes \ell}\otimes I_{n+m}^{\otimes t-\ell}\right) \right].
\]
Then $A_{\bfv,\sigma} = \ketbra{\bfv}{\sigma(\bfv)}$ if $\sigma$ maps $[\ell]$ to $[\ell]$; otherwise, $A_{\bfv,\sigma} = 0$. 
\end{lemma}
\begin{proof}
Suppose $\bfv = (v_1||w_1,\ldots,v_t||w_t)\in\bit^{(n+m)t}$ with $v_i\in\bit^n$ and $w_i\in\bit^m$ for all $i\in[t]$. First, a direct calculation yields
\[
\left(\left(Z^{k}\otimes I_m\right)^{\otimes \ell}\otimes I_{n+m}^{\otimes t-\ell}\right)\ketbra{\bfv}{\sigma(\bfv)}\left(\left(Z^{k}\otimes I_m\right)^{\otimes \ell}\otimes I_{n+m}^{\otimes t-\ell}\right) 
= (-1)^{\langle k,\oplus_{i=1}^{\ell} (v_i\oplus v_{\sigma(i)})\rangle}\ketbra{\bfv}{\sigma(\bfv)}.
\]
Therefore, after averaging over $k$,
\[
A_{\bfv,\sigma} 
= \E_{k\in\bit^{n}}\left[(-1)^{\langle k,\oplus_{i=1}^{\ell} (v_i\oplus v_{\sigma(i)})\rangle}\right] \ketbra{\bfv}{\sigma(\bfv)}
= \begin{cases}
\ketbra{\bfv}{\sigma(\bfv)} & \text{ if } \oplus_{i=1}^{\ell} (v_i\oplus v_{\sigma(i)}) = 0 \\
0 & \text{ otherwise.}
\end{cases}
\]
Since $\type(\bfv) \in \goodpre{n}{m}{\ell}{t}$, the condition $\oplus_{i=1}^{\ell} v_i  = \oplus_{i=1}^\ell v_{\sigma(i)}$ holds if and only if the two sets $\set{1,2,\dots,\ell}$ and $\set{\sigma(1),\sigma(2),\dots,\sigma(\ell)}$ are identical. 
\end{proof}

\noindent The following lemma lies in the technical heart of this work. It says that the action of applying random $Z^k$ on $\ell$-fold $n$-prefix collision-free types $T$ has the following ``classical'' probabilistic interpretation: the output is identically distributed to first uniformly sampling an $\ell$-subset $X\subset T$ and then outputting $\ketbra{X}{X} \otimes \ketbra{T\setminus X}{T\setminus X}$.
\begin{lemma} \label{lem:nice_T}
For any $T \in \goodpre{n}{m}{\ell}{t}$,
\[
\Ex_{k\in\bit^{n}}\left[ \left( \left(Z^{k}\otimes I_m\right)^{\otimes \ell}\otimes I_{n+m}^{\otimes t-\ell} \right) \ketbra{T}{T} \left(\left(Z^{k}\otimes I_m\right)^{\otimes \ell}\otimes I_{n+m}^{\otimes t-\ell}\right) \right] 
= \Ex_{X\subset T} \left[ \ketbra{X}{X} \otimes \ketbra{T\setminus X}{T\setminus X} \right],
\]
where $X$ is a uniform $\ell$-subset of $T$.\footnote{Since $T$ is collision-free, we can interpret it as a set.}
\end{lemma}
\begin{proof}
We first use the expression in~\Cref{eq:useful} on the left-hand side:
\begin{align} \label{eq:split}
LHS 
= \Ex_{\bfv\gets T} \left[ \sum_{\sigma\in S_t} \Ex_{k\in\bit^{n}}\left[ \left( \left(Z^{k}\otimes I_m\right)^{\otimes \ell}\otimes I_{n+m}^{\otimes t-\ell} \right) \ketbra{\bfv}{\sigma(\bfv)} \left( \left(Z^{k}\otimes I_m\right)^{\otimes \ell}\otimes I_{n+m}^{\otimes t-\ell} \right) \right] \right].
\end{align}
Then from the previous lemma (\Cref{lem:perm_split})
\begin{align*}
\eqref{eq:split}
& = \Ex_{\bfv\gets T} \left[ \sum_{\sigma_1 \in S_\ell, \sigma_2 \in S_{t-\ell}} \ketbra{\bfv}{\sigma_1\circ \sigma_2(\bfv)} \right] \\
& = \Ex_{\bfv\gets T} \left[ \sum_{\sigma_1 \in S_\ell} \ketbra{\bfv_{[1:\ell]}}{\sigma_1(\bfv_{[1:\ell]})} \otimes \sum_{\sigma_2 \in S_{t-\ell}} \ketbra{\bfv_{[\ell+1:t]}}{\sigma_2(\bfv_{[\ell+1:t]})} \right] \\
& = \Ex \left[ \sum_{\sigma_1 \in S_\ell} \ketbra{\bfv_1}{\sigma_1(\bfv_1)} \otimes \sum_{\sigma_2 \in S_{t-\ell}} \ketbra{\bfv_2}{\sigma_2(\bfv_2)}: \substack{X\subset T, \\ \bfv_1\gets X, \\ \bfv_2\gets T\setminus X} \right] \\
& = \Ex_{X\subset T} \left[ \ketbra{X}{X} \otimes \ketbra{T\setminus X}{T\setminus X} \right].
\end{align*}
For the first equality, we use~\Cref{lem:perm_split} and decompose $\sigma = \sigma_1 \circ \sigma_2$ for some $\sigma_1,\sigma_2$ such that $\sigma_1(x) = x$ for all $x \in \set{\ell+1,\ell+2,\cdots,t}$ and $\sigma_2(y) = y$ for all $y \in \set{1,2,\cdots,\ell}$, which uniquely correspond to elements in $S_\ell$ and $S_{t-\ell}$. The second equality follows by denoting the first $\ell$ part of $\bfv$ by $\bfv_{[1:\ell]}$ and the last $t-\ell$ part of $\bfv$ by $\bfv_{[\ell+1:t]}$. The third equality holds because sampling $\bfv$ from $T$ is equivalent to sampling $X\subset T$ followed by assigning order to elements in $X$ and $T\setminus X$.
\end{proof}
\section{$\ell$-copy statistical PRS in the CHS model}
In this section, we discuss the construction of $\ell$-copy statistical PRS in the CHS model for $\ell = O(\secp/\log(\secp)^{1+\epsilon})$ (for any constant $\epsilon > 0$) and length of the common Haar state $n\geq \secp$. We also show that the construction satisfies multikey security for the same parameters. Further, we complement our result by showing that for $n = \omega(\log(\secp))$, achieving an $\ell$-copy statistical PRS in the CHS model for $\ell = \Omega(\secp/\log(\secp))$ is impossible (similarly for the multikey case), this shows that our construction is optimal (best one can hope for) for $n\geq \secp$.

\begin{theorem}
For all constants $\epsilon>0$, there exists multi-key $O(\secp/\log(\secp)^{1+\epsilon})$-copy statistical stretch PRS in the CHS model. 
\end{theorem}

\subsection{Construction}\label{sec:prs-con}
In this section, we assume that $n(\secp) \geq \secp$. We define the construction as follows: on input $k\in\set{0,1}^{\secp}$ of the $n$-qubit common Haar state $\ket{\vartheta}$,
\[
G_k(\ket{\vartheta}) := (Z^{k}\otimes I_{n-\secp})\ket{\vartheta}.
\]
Note that when $n(\secp)>\secp$, our construction is a \emph{stretch} PRS. 
\begin{lemma}[Pseudorandomness] \label{lem:prs-sec}
Let $G$ be as defined above. Let 
\[
\rho := \E_{\substack{k\gets\set{0,1}^{\secp} \\ \ket{\vartheta}\gets\Haar_{n}}}
\left[
G_k(\ket{\vartheta})^{\otimes \ell}\otimes\ketbra{\vartheta}{\vartheta}^{\otimes t}
\right]
\text{ and }
\sigma := \Ex_{\substack{\ket{\varphi}\gets\Haar_n \\ \ket{\vartheta}\gets\Haar_{n}}}
\left[
\ketbra{\varphi}{\varphi}^{\otimes \ell}\otimes\ketbra{\vartheta}{\vartheta}^{\otimes t}
\right].
\]
Then $\TD\left(\rho,\sigma\right) = O\left(\frac{(\ell+t)^{2\ell}}{2^{\secp}}\right)$.
\end{lemma}
\begin{proof}
We prove this using hybrid arguments:
\paragraph{Hybrid $1$.} 
Sample $T\gets [0:\ell+t]^{2^n}$. 
Sample $k \gets \set{0,1}^{\secp}$. 
Output $((Z^{k} \otimes I_{n-\secp})^{\otimes \ell}\otimes I_{n}^{\otimes t})\ket{T}$.
\paragraph{Hybrid $2$.} 
Sample $T\gets [0:\ell+t]^{2^n}$ uniformly conditioned on $T\in\goodpre{\secp}{n-\secp}{\ell}{\ell+t}$. 
Sample $k\gets\set{0,1}^{\secp}$. 
Output $((Z^{k}\otimes I_{n-\secp})^{\otimes \ell}\otimes I_{n}^{\otimes t})\ket{T}$.
\paragraph{Hybrid $3$:}
Sample $T\gets [0:\ell+t]^{2^n}$ uniformly conditioned on $T\in\goodpre{\secp}{n-\secp}{\ell}{\ell+t}$.
Sample a uniform $\ell$-subset $T_1$ from $T$.
Output $\ket{T_1}\otimes\ket{T\setminus T_1}$.
\paragraph{Hybrid $4$.} 
Sample $T\gets [0:\ell+t]^{2^n}$.
Sample a uniform $\ell$-subset $T_1$ from $T$.\footnote{Since $T$ might have collision, $T_1$ is a multiset of size $\ell$.}
Output $\ket{T_1}\otimes\ket{T\setminus T_1}$.
\paragraph{Hybrid $5$.} 
Sample a collision-free $T$ from $[0:\ell+t]^{2^n}$.
Sample a uniform $\ell$-subset $T_1$ from $T$.
Output $\ket{T_1}\otimes\ket{T\setminus T_1}$.
\paragraph{Hybrid $6$.} 
Sample a uniform collision-free $T_1$ from $[0:\ell]^{2^n}$.
Sample a uniform collision-free $T_2$ from $[0:t]^{2^n}$ conditioned on $T_1$ and $T_2$ have no common elements.
Output $\ket{T_1} \otimes \ket{T_2}$.
\paragraph{Hybrid $7$.} 
Sample a uniform collision-free $T_1$ from $[0:\ell]^{2^n}$.
Sample a uniform collision-free $T_2$ from $[0:t]^{2^n}$.
Output $\ket{T_1} \otimes \ket{T_2}$.
\paragraph{Hybrid $8$.} 
Sample $T_1\gets [0:\ell]^{2^n}$.
Sample $T_2\gets [0:t]^{2^n}$.
Output $\ket{T_1} \otimes \ket{T_2}$.

\noindent By~\Cref{fact:random_type_l_fold_collision}, the trace distance between Hybrid~$1$ and Hybrid~$2$ is $O((t+\ell)^{2\ell}/2^{\secp})$.  
From~\Cref{lem:nice_T}, the output of Hybrid~$2$ is
\[
\Ex_{\substack{T \gets [0:\ell+t]^{2^n}: \\ T\in\goodpre{\secp}{n-\secp}{\ell}{\ell+t} }}
\Ex_{T_1 \subset T} \left[ \ketbra{T_1}{T_1} \otimes \ketbra{T\setminus T_1}{T\setminus T_1} \right].
\]
Hence, Hybrid~$2$ is equivalent to Hybrid~$3$. From now on, we will prove the closeness of the remaining hybrids by classical arguments. Again by~\Cref{fact:random_type_l_fold_collision}, the trace distance between Hybrid~$3$ and Hybrid~$4$ is $O((t+\ell)^{2\ell}/2^{\secp})$. The trace distance between Hybrid~$4$ and Hybrid~$5$ is $O((t+\ell)^{2}/2^{n})$ by collision bound. Hybrid~$5$ and Hybrid~$6$ are equivalent. The trace distance between Hybrid~$6$ and Hybrid~$7$ is $O(t\ell/2^{n})$. Finally, the trace distance between Hybrid~$7$ and Hybrid~$8$ is $O((t^2+\ell^2)/2^{n})$ by collision bound. Collecting the probabilities completes the proof.

\end{proof}

\begin{remark}
From~\Cref{lem:prs-sec}, we conclude that the construction given above is an $\ell$-copy statistical PRS in the CRHS model for $\ell = O(\secp/\log(\secp)^{1+\epsilon})$ (for any constant $\epsilon > 0$). 
\end{remark}

\noindent By the following lemma, we can also show that the above construction is a multi-key $\ell$-copy statistical PRS in the CRHS model for $\ell = O(\secp/\log(\secp)^{1+\epsilon})$ (for any constant $\epsilon > 0$): 

\begin{lemma}[Multi-key pseudorandomness]\label{lem:prs-multi-sec}
    Let $G$ be as defined above. Let 
\[
\rho := 
\bigotimes_{i=1}^p \Ex_{\ket{\varphi_i} \gets \Haar_n} 
\left[ 
\ketbra{\varphi_i}{\varphi_i}^{\otimes \ell} 
\right]
\otimes
\Ex_{\ket{\vartheta}\gets\Haar_{n}}
\left[
\ketbra{\vartheta}{\vartheta}^{\otimes t}
\right]
\text{ and }
\sigma := \Ex_{\ket{\vartheta}\gets\Haar_{n}}
\left[
\bigotimes_{i=1}^p  \Ex_{k_i\gets \bit^\secp} 
\left[ 
G_{k_i}(\ket{\vartheta})^{\otimes \ell} 
\right]
\otimes \ketbra{\vartheta}{\vartheta}^{\otimes t}
\right].
\]
Then $\TD\left(\rho,\sigma\right) = O\left(\frac{p \cdot (p\ell+t)^{2\ell}}{2^{\secp}}\right)$.
\end{lemma}
\begin{proof}
For $j = 0,1,\dots,p$, we define the following (hybrid) density matrices:\footnote{Similar to proving the output of a classical PRG on polynomial i.i.d uniform keys is computationally indistinguishable from polynomial i.i.d uniform strings, we can construct a security reduction to simulate these hybrids. However, since we are in the information-theoretic setting, we instead calculate their trace distances directly.}
\begin{align*}
\xi_j :=
\bigotimes_{i=1}^j \Ex_{\ket{\varphi_i}\gets\Haar_{n}}
\left[ 
\ketbra{\varphi_i}{\varphi_i}^{\otimes \ell} 
\right]
\otimes
\Ex_{\ket{\vartheta}\gets\Haar_{n}}
\left[
\bigotimes_{i=j+1}^p \Ex_{k_i \gets \bit^\secp} 
\left[ 
G_{k_i}(\ket{\vartheta})^{\otimes \ell} 
\right]
\otimes \ketbra{\vartheta}{\vartheta}^{\otimes t}
\right].
\end{align*}
We will complete the poof by showing that $\TD(\xi_j,\xi_{j+1}) = O\left( \frac{ ((p-j) \cdot \ell + t)^{2\ell} }{ 2^{\secp} } \right)$ for $j = 0,1,\dots,p-1$. By the property that $\TD(A\otimes X, A\otimes Y) = \TD(X,Y)$, the trace distance between $\xi_j$ and $\xi_{j+1}$ is identical to that of 
\[
\xi'_j :=
\Ex_{\ket{\vartheta}\gets\Haar_{n}}\left[ \bigotimes_{i=j+1}^p \Ex_{k_i \gets \bit^\secp} [G_{k_i}(\ket{\vartheta})^{\otimes \ell}] \otimes \ketbra{\vartheta}{\vartheta}^{\otimes t} \right]
\]
and
\[
\xi'_{j+1} :=
\Ex_{\ket{\varphi_{j+1}}\gets\Haar_{n}} 
\left[ 
\ketbra{\varphi_{j+1}}{\varphi_{j+1}}^{\otimes \ell} 
\right]
\otimes
\Ex_{\ket{\vartheta}\gets\Haar_{n}}\left[ \bigotimes_{i=j+2}^p \Ex_{k_i \gets \bit^\secp} [G_{k_i}(\ket{\vartheta})^{\otimes \ell}] \otimes \ketbra{\vartheta}{\vartheta}^{\otimes t} \right].
\]
By the monotonicity of trace distance (\ie $\TD(\cE(X),\cE(Y))\leq \TD(X,Y)$ for any quantum channel $\cE$) and setting $\cE := \bigotimes_{i=j+2}^p \Ex_{k_i \gets \bit^\secp} [G_{k_i}(\cdot)^{\otimes \ell}]$,\footnote{The channel $\cE$ acts as the identity on unspecified registers.} we have
\begin{align*}
& \TD(\xi'_j,\xi'_{j+1}) \leq \\
& \TD\left(
\Ex_{ \substack{k_{j+1} \gets \bit^\secp, \\ \ket{\vartheta} \gets \Haar_{n}} }
\left[ G_{k_{j+1}}(\ket{\vartheta})^{\otimes \ell} 
\otimes \ketbra{\vartheta}{\vartheta}^{\otimes (p-j-1)\ell+t} \right], 
\Ex_{ \substack{\ket{\varphi_{j+1}}\gets\Haar_{n}, \\ \ket{\vartheta} \gets \Haar_{n}} }
\left[ \ketbra{\varphi_{j+1}}{\varphi_{j+1}}^{\otimes \ell} 
\otimes \ketbra{\vartheta}{\vartheta}^{\otimes (p-j-1)\ell+t} \right]
\right) \\
& = O\left( \frac{ \left( (p-j) \cdot \ell + t \right)^{2\ell} }{ 2^{\secp} } \right),
\end{align*}
where the last equality follows from~\Cref{lem:prs-sec}. Applying the triangle inequality completes the proof.
\end{proof}

\noindent As a remark, note that~\Cref{lem:prs-sec} gives a simpler proof of the following theorem in~\cite{Col23}:
\begin{lemma}[{\cite[Lemma~5]{Col23}}] \label{lem:Z_haar_indis}
Consider the ensemble of states:
\[
\set{\rho_x}_{x \in \bit^n} 
= \left\{ \Ex_{\ket{\psi}\gets \Haar_n} \left[ (Z^x \otimes I^{\otimes m}) \ketbra{\psi}{\psi}^{\otimes m+1} (Z^x \otimes I^{\otimes m}) \right] \right\}_{x \in \bit^n}.
\]
Then, there is a constant $C > 0$, such that, for any POVM $\set{M_x}_{x \in \bit^n}$,
\[
\Ex_{x\gets\bit^n} \Tr( M_x\rho_x ) = C \cdot \left( \frac{m}{2^n} + \frac{m^7}{2^{3n}} \right)^{\frac{1}{2}}.
\]
\end{lemma}
\noindent By setting $\ell = 1, t = m, n = \secp$ in~\Cref{lem:prs-sec} and a hybrid, the success probability is at most $O(m^2/2^n)$ plus the probability of inverting $x$ when the input is $\sigma$, which is at most $1/2^n$ since $\sigma$ is independent of $x$.

In~\Cref{app:simpleproof}, we further give another proof by simplifying the calculation in~\Cref{lem:Z_haar_indis}, which may be of independent interest. Moreover, we eliminate the $m^7/2^{3n}$ term.

\subsection{Impossibility of a special class of PRS in CHS model}
\label{sec:prs-imp}

In this subsection, if the PRS generation algorithm uses only \emph{one copy} of the common Haar state, we show that $\ell$-copy statistical PRS and multi-key $\ell$-copy statistical PRS are impossible for $\ell = \Omega(\secp/\log(\secp))$ and $n=\omega(\log(\secp))$.

\begin{theorem}\label{thm:prs-imp}
$\ell$-copy statistical PRS is impossible in the CHS model if (a) the generation algorithm uses only one copy of the common Haar state, (b) $\ell = \Omega(\secp/\log(\secp))$ and (c) the length of the common Haar state is $\omega(\log(\secp))$.
\end{theorem}

\begin{proof}
We prove this by showing that for $t(\secp) := \secp^3$ and $\ell(\secp) := \secp/\log(\secp)$, there exists a (computationally unbounded) adversary $\cA$ such that
\[
\left|
\Pr_{\substack{k\gets\set{0,1}^{\secp}\\ \ket{\vartheta}\gets\Haar_{n}}}[\cA(\ketbra{\vartheta}{\vartheta}^{\otimes t}\otimes G(k,\ketbra{\vartheta}{\vartheta})^{\otimes \ell}) = 1] 
- \Pr_{\substack{\ket{\varphi}\gets\Haar_{n}\\ \ket{\vartheta}\gets\Haar_{n}}}[\cA(\ketbra{\vartheta}{\vartheta}^{\otimes t}\otimes \ketbra{\varphi}{\varphi}^{\otimes \ell}) = 1]
\right|
\]
is non-negligible. For short, we use the following notation:
\[
\rho_0 
:= \Ex_{k\gets\set{0,1}^{\secp}, \ket{\vartheta}\gets\Haar_{n}}\left[ \ketbra{\vartheta}{\vartheta}^{\otimes t}\otimes G(k,\ketbra{\vartheta}{\vartheta})^{\otimes \ell} \right]
\]
and
\[
\rho_1 
:= \Ex_{\ket{\varphi}\gets\Haar_{n}, \ket{\vartheta}\gets\Haar_{n}}
\left[ \ketbra{\vartheta}{\vartheta}^{\otimes t}\otimes \ketbra{\varphi}{\varphi}^{\otimes \ell} \right].
\]
The adversary $\cA$ is simple: it performs a binary measurement $\set{\Pi, I-\Pi}$ on input $\rho_b$ for $b\in\bit$, where $\Pi$ is the projection onto the eigenspace of $\rho_0$. The rank of $\rho_0$ and $\rho_1$ satisfies
\[
\rank(\rho_0) 
\leq 2^{\secp} \cdot {2^n + \ell+t-1 \choose \ell+t}
\quad \text{and} \quad
\rank(\rho_1) 
= {2^n + \ell-1 \choose \ell} \cdot {2^n + t-1 \choose t}.
\]
Now, by construction, we have 
\[
\Pr_{\substack{k\gets\set{0,1}^{\secp}\\ \ket{\vartheta}\gets\Haar_{n}}}[\cA(\ketbra{\vartheta}{\vartheta}^{\otimes t}\otimes G(k,\ketbra{\vartheta}{\vartheta})^{\otimes \ell}) = 1]
= \Tr(\Pi\rho_0) 
= \Tr(\rho_0) 
= 1.
\]
On the other hand, suppose $\Pi = \sum_{i=1}^{\rank(\rho_0)} \ketbra{u_i}{u_i}$, then
\begin{align*}
& \Pr_{\substack{\ket{\varphi}\gets\Haar_{n}\\ \ket{\vartheta}\gets\Haar_{n}}}[\cA(\ketbra{\vartheta}{\vartheta}^{\otimes t}\otimes \ketbra{\varphi}{\varphi}^{\otimes \ell}) = 1]
= \Tr(\Pi\rho_1) \\
& \leq \sum_{i=1}^{\rank(\rho_0)} 
\frac{1}{\binom{2^n + \ell - 1}{\ell}\binom{2^n + t - 1}{t}} \cdot
\sum_{T_1 \in [0:\ell]^{2^n}, T_2 \in [0:t]^{2^n}} |(\bra{T_1}\otimes\bra{T_2})\ket{u_i}|^2 \\
& \leq \frac{\rank(\rho_0)}{\binom{2^n + \ell - 1}{\ell}\binom{2^n + t - 1}{t}}
= \frac{\rank(\rho_0)}{\rank(\rho_1)}.
\end{align*}
A direct calculation yields:
\begin{align*}
\frac{\rank(\rho_0)}{\rank(\rho_1)} 
& = \frac{2^\secp}{\binom{\ell+t}{\ell}} \cdot
\prod_{i = 0}^{\ell-1} \left(1 + \frac{t}{2^n + i}\right)
\leq \frac{2^\secp}{(1+\frac{t}{\ell})^\ell} \cdot \prod_{i = 0}^{\ell-1} \left(1 + \frac{t}{2^n + i}\right) \\
& = 2^\secp \cdot \prod_{i = 0}^{\ell-1} \left( \frac{1 + \frac{t}{2^n + i}}{1+\frac{t}{\ell}} \right) 
\leq 2^\secp \cdot \left( \frac{1 + \frac{t}{2^n}}{1+\frac{t}{\ell}} \right)^\ell,
\end{align*}
where the first inequality follows from $\binom{\ell+t}{\ell} \geq (\frac{\ell+t}{\ell})^\ell$.
For $n = \omega(\log(\secp)), t = \secp^3$ and $\ell = \secp/\log(\secp)$, we have
\[
2^\secp \cdot \left( \frac{1 + \frac{t}{2^n}}{1+\frac{t}{\ell}} \right)^\ell
= \left( \frac{\secp \cdot (1 + \frac{\secp^3}{\secp^{\omega(1)}})}{1+\secp^2\log(\secp)}  \right)^{\secp/\log(\secp)} \leq \left( \frac{\secp \cdot 2}{\secp^2\log(\secp)}  \right)^{\secp/\log(\secp)}
\leq 2^{-\secp}
\]
for sufficiently large $\secp$. Hence, the distinguishing advantage ($1-2^{\secp}$) is non-negligible. This completes the proof.
\end{proof}

\noindent Note that since $\ell$-copy statistical PRS in the CHS model implies multi-key $\ell$-copy statistical PRS in the CHS model, we also have the following:

\begin{theorem}
Multi-key $\ell$-copy statistical PRS is impossible in the CHS model if (a) the generation algorithm uses only one copy of the common Haar state, (b) $\ell = \Omega(\secp/\log(\secp))$ and (c) the length of the common Haar state is $\omega(\log(\secp))$.
\end{theorem}

\begin{proof}
    Since $\ell$-copy statistical PRS in the CHS model implies multi-key $\ell$-copy statistical PRS in the CHS model, by~\Cref{thm:prs-imp}, multi-key $\ell$-copy statistical PRS is impossible in CHS model for $\ell = \Omega(\secp/\log(\secp))$.
\end{proof}

\section{Quantum commitments in CHS model} \label{sec:Commitment}
In this section, we construct a commitment scheme that satisfies poly-copy statistical hiding and statistical sum-biding in the CHS model. The scheme is inspired by the quantum commitment scheme proposed in~\cite{MY21}, adapting it to the CHS model and providing a formal proof of its security. In contrast to the scheme in~\cite{MY21}, our construction is not in the canonical form~\cite{Yan22}. Hence, we need SWAP tests to verify rather than uncomputing, this makes the proof of binding more involved. Our scheme relies on the multi-key pseudorandomness property of the construction giving in~\Cref{sec:prs-con} for hiding. 

\subsection{Construction}
We assume that $n(\secp) \geq \secp + 1$ for all $\secp\in\N$. Our construction is shown in~\Cref{fig:commitment}.

\begin{figure}[h]
   \begin{tabular}{|p{16cm}|}
   \hline \\
\par Commit phase: The committer $C_{\secp}$ on input $b\in\bit$ does the following:
    \begin{itemize}
        \item Use $p(\secp) = \secp$ copies of the common Haar state $\ket{\vartheta}$ to prepare the state $\ket{\Psi_b}_{\sfC\sfR} := \bigotimes_{i=1}^p \ket{\psi_b}_{\sfC_i\sfR_i}$, where 
        \[
            \ket{\psi_0}_{\sfC_i\sfR_i} 
            := \frac{1}{\sqrt{2^\secp}} \sum_{k \in \bit^\secp} (Z^k \otimes I_{n-\secp}) \ket{\vartheta}_{\sfC_i} \ket{k||0^{n-\secp}}_{\sfR_i}
        \]
        and
        \[
            \ket{\psi_1}_{\sfC_i\sfR_i} 
            := \frac{1}{\sqrt{2^n}} \sum_{j \in \bit^n} \ket{j}_{\sfC_i} \ket{j}_{\sfR_i},
        \]
        and $\sfC := (\sfC_1, \sfC_2, \dots, \sfC_p)$ and $\sfR := (\sfR_1, \sfR_2, \dots, \sfR_p)$.
        \item Send the register $\sfC$ to the receiver.
    \end{itemize}

\par Reveal phase: 
\begin{itemize}
    \item The committer sends $b$ and the register $\sfR$ to the receiver.
    \item The receiver prepares the state $\ket{\Psi_b}_{\sfC'\sfR'} = \bigotimes_{i=1}^p \ket{\psi_b}_{\sfC'_i\sfR'_i}$ by using $p$ copies of the common Haar state $\ket{\vartheta}$, where $\sfC' := (\sfC'_1, \sfC'_2, \dots, \sfC'_p)$ and $\sfR' := (\sfR'_1, \sfR'_2, \dots, \sfR'_p)$ are receiver's registers.
    \item For $i\in[p]$, the receiver performs the SWAP test between registers $(\sfC_i,\sfR_i)$ and $(\sfC'_i,\sfR'_i)$.
    \item The receiver outputs $b$ if all SWAP tests accept; otherwise, outputs $\bot$.
\end{itemize}
\\ 
\hline
\end{tabular}
\caption{Commitment scheme in the CHS model}
\label{fig:commitment}
\end{figure}

\subsection{Proof of Correctness, Hiding, and Binding}
Clearly, the construction given in~\Cref{fig:commitment} has perfect correctness.
\begin{theorem} \label{thm:com}
The construction given in~\Cref{fig:commitment} satisfies poly-copy statistical hiding and statistical sum binding.
\end{theorem}
\begin{proof}[Proof of~\Cref{thm:com}]
\item \paragraph{Poly-copy statistical hiding.} 
It follows immediately from~\Cref{lem:prs-multi-sec} by setting $\ell = 1$.

\item \paragraph{Statistical sum binding.}
For any (fixed) common Haar state $\ket{\vartheta}$ and $i\in[p]$, it holds that
\begin{align} \label{eq:fidelity}
& F( \Tr_{\sfR_i}( \ketbra{\psi_0}{\psi_0}_{\sfC_i\sfR_i}), \Tr_{\sfR_i}( \ketbra{\psi_1}{\psi_1}_{\sfC_i\sfR_i}) ) \nonumber \\
= & F\left( \underbrace{\frac{1}{2^\secp} \sum_{k \in \bit^\secp} (Z^k \otimes I_{n-\secp}) \ketbra{\vartheta}{\vartheta}_{\sfC_i}(Z^k \otimes I_{n-\secp})}_{\rho_0}, \frac{I_{\sfC_i} }{2^n} \right) \nonumber \\
= & 2^{-n} \cdot \Tr( \sqrt{\rho_0} )^2 \nonumber \\
\leq & 2^{-n} \cdot \rank(\sqrt{\rho_0}) \cdot \Tr(\rho_0) \nonumber \\
\leq & 2^{-n} \cdot 2^{\secp} \cdot 1
= 2^{-(n-\secp)},
\end{align}
where the second equality is by the definition of fidelity $F(\rho,\sigma) = \left( \Tr(\sqrt{\sqrt{\rho}\sigma\sqrt{\rho}}) \right)^2$; the first inequality follows from $\Tr(\rho)^2 \leq \rank(\rho) \cdot \Tr(\rho^2)$ for $\rho \succeq 0$; the second inequality is because $\rank(\sqrt{\rho}) = \rank(\rho)$ for $\rho \succeq 0$ and $\rank(X+Y) \leq \rank(X) + \rank(Y)$.

\noindent Let $M^{(b)}_{\sfC\sfR}$ be the POVM operator corresponding to that the receiver outputs $b$ (\ie all the SWAP tests accept), 
\[
M^{(b)}_{\sfC\sfR} := \bigotimes_{i\in[p]} \left( \frac{I_{\sfC_i\sfR_i} + \ketbra{\psi_b}{\psi_b}_{\sfC_i\sfR_i})}{2} \right)
= \Ex_{\cS\subseteq[p]} \left[ \bigotimes_{i\in \cS}\ketbra{\psi_b}{\psi_b}_{\sfC_i\sfR_i} \otimes \bigotimes_{i\notin \cS} I_{\sfC_i\sfR_i} \right],
\]
where $\cS$ is a uniformly random subset of $[p]$. Then the probability that the receiver outputs $b$ is
\begin{align*}
p_b & := \Tr\left(M^{(b)}_{\sfC\sfR} \underbrace{\Tr_\sfE( U^{(b)}_{\sfR\sfE} \ketbra{\Phi}{\Phi}_{\sfC\sfR\sfE} U^{(b)\dagger}_{\sfR\sfE}}_{\rho^{(b)}_{\sfC\sfR}} ) \right) \\
& = \Ex_{\cS\subseteq[p]} \left[ \Tr\left( \bigotimes_{i\in \cS}\ketbra{\psi_b}{\psi_b}_{\sfC_i\sfR_i} \otimes \bigotimes_{i\notin \cS} I_{\sfC_i\sfR_i} \cdot \rho^{(b)}_{\sfC\sfR} \right) \right] \\
& = \Ex_{\cS\subseteq[p]} \left[ \underbrace{ F \left( \bigotimes_{i\in \cS}\ketbra{\psi_b}{\psi_b}_{\sfC_i\sfR_i}, \Tr_{\sfC_i\sfR_i: i\notin\cS}(\rho^{(b)}_{\sfC\sfR}) \right)}_{p_{b,S}} \right],
\end{align*}
where $\sfE$ is the committer's internal register, $\ket{\Phi}_{\sfC\sfR\sfE}$ is the malicious committer's initial state that might depend on $\ket{\vartheta}$ (we omit the dependence for simplicity), and $U^{(b)}_{\sfR\sfE}$ is the malicious committer's attacking unitary for $b$; we plug in the definition of $M^{(b)}_{\sfC\sfR}$ and use the short-hand notation $\rho^{(b)}_{\sfC\sfR}$ to obtain the second equality.

\noindent For any fixed $\cS \subseteq [p]$, we have
\begin{align*}
& p_{0,\cS} + p_{1,\cS} \\
& = F\left( \bigotimes_{i\in \cS}\ketbra{\psi_0}{\psi_0}_{\sfC_i\sfR_i}, \Tr_{\sfC_i\sfR_i: i\notin\cS}(\rho^{(0)}_{\sfC\sfR}) \right) 
+ F\left( \bigotimes_{i\in \cS}\ketbra{\psi_1}{\psi_1}_{\sfC_i\sfR_i}, \Tr_{\sfC_i\sfR_i: i\notin\cS}(\rho^{(1)}_{\sfC\sfR}) \right) \\
& \leq F\left( \bigotimes_{i\in \cS} \Tr_{\sfR_i} (\ketbra{\psi_0}{\psi_0}_{\sfC_i\sfR_i}), \Tr_{\sfC_i:i\notin \cS}\Tr_{\sfR}(\rho^{(0)}_{\sfC\sfR}) \right) 
+ F\left( \bigotimes_{i\in \cS} \Tr_{\sfR_i} (\ketbra{\psi_1}{\psi_1}_{\sfC_i\sfR_i}), \Tr_{\sfC_i:i\notin \cS}\Tr_{\sfR}(\rho^{(1)}_{\sfC\sfR}) \right) \\
& \leq 1 + F\left( \bigotimes_{i\in \cS} \Tr_{\sfR_i} (\ketbra{\psi_0}{\psi_0}_{\sfC_i\sfR_i}), \bigotimes_{i\in \cS} \Tr_{\sfR_i} (\ketbra{\psi_1}{\psi_1}_{\sfC_i\sfR_i}) \right)^{1/2} \\
& = 1 + \bigotimes_{i\in \cS} F \left( \Tr_{\sfR_i} (\ketbra{\psi_0}{\psi_0}_{\sfC_i\sfR_i}), \Tr_{\sfR_i} (\ketbra{\psi_1}{\psi_1}_{\sfC_i\sfR_i}) \right)^{1/2} 
\leq 1 + 2^{\frac{-|\cS|(n-\secp)}{2}},
\end{align*}
where the first inequality follows from the fact that taking a partial trace won't decrease the fidelity; the second inequality is because $\Tr_{\sfR}(\rho^{(0)}_{\sfC\sfR}) = \Tr_{\sfR}(\rho^{(1)}_{\sfC\sfR})$ and $F(\rho,\xi)+F(\sigma,\xi) \leq 1 + \sqrt{F(\rho,\sigma)}$~\cite{NS03}; the last equality follows from the fact that $F(\bigotimes_i \rho_i, \bigotimes_i \sigma_i) = \prod_i F(\rho_i, \sigma_i)$; the last inequality follows from~\Cref{eq:fidelity}. Finally, we bound the probability $p_0 + p_1$ as follows:
\begin{align*}
p_0 + p_1
& = \Ex_{\cS\in[p]} \left[ p_{0,\cS} + p_{1,\cS} \right] 
\leq 1 + \Ex_{\cS\in[p]} \left[ 2^{\frac{-|\cS|(n-\secp)}{2}} \right]
= 1 + 2^{-p} \cdot \sum_{s = 0}^t \binom{p}{s} 2^{\frac{-s(n-\secp)}{2}} \\
& = 1 + \left( \frac{ 1 + 2^{\frac{-(n-\secp)}{2}} }{2} \right)^p
= 1 + \negl(\secp). \qedhere
\end{align*}
\end{proof}

\section*{Acknowledgements}
This work is supported by the National Science Foundation under Grant No. 2329938 and Grant No. 2341004.

\printbibliography

@inproceedings{Yan22,
  title={General properties of quantum bit commitments},
  author={Yan, Jun},
  booktitle={International Conference on the Theory and Application of Cryptology and Information Security},
  pages={628--657},
  year={2022},
  organization={Springer}
}

@article{LoChau97,
  title={Is quantum bit commitment really possible?},
  author={Lo, Hoi-Kwong and Chau, Hoi Fung},
  journal={Physical Review Letters},
  volume={78},
  number={17},
  pages={3410},
  year={1997},
  publisher={APS}
}

@inproceedings{Bra23,
  title={Black-Hole Radiation Decoding Is Quantum Cryptography},
  author={Brakerski, Zvika},
  booktitle={Annual International Cryptology Conference},
  pages={37--65},
  year={2023},
  organization={Springer}
}

@article{BCQ23,
  title={On the computational hardness needed for quantum cryptography},
  author={Brakerski, Zvika and Canetti, Ran and Qian, Luowen},
  journal={arXiv preprint arXiv:2209.04101},
  year={2022}
}

@article{May97,
  title={Unconditionally secure quantum bit commitment is impossible},
  author={Mayers, Dominic},
  journal={Physical review letters},
  volume={78},
  number={17},
  pages={3414},
  year={1997},
  publisher={APS}
}

@article{Qian23,
  title={Unconditionally secure quantum commitments with preprocessing},
  author={Qian, Luowen},
  journal={Cryptology ePrint Archive},
  year={2023}
}

@misc{MY21,
  doi = {10.48550/ARXIV.2112.06369},
  
  author = {Morimae, Tomoyuki and Yamakawa, Takashi},
  
  title = {Quantum commitments and signatures without one-way functions},
  
  publisher = {arXiv},
  
  year = {2021},
  
  copyright = {Creative Commons Attribution 4.0 International}
}

@misc{MY23onewayness,
      title={One-Wayness in Quantum Cryptography}, 
      author={Tomoyuki Morimae and Takashi Yamakawa},
      year={2023},
      eprint={2210.03394},
      archivePrefix={arXiv},
      primaryClass={quant-ph}
}

@inproceedings{AGQY22,
  title={Pseudorandom (Function-Like) Quantum State Generators: New Definitions and Applications},
  author={Ananth, Prabhanjan and Gulati, Aditya and Qian, Luowen and Yuen, Henry},
  booktitle={Theory of Cryptography Conference},
  pages={237--265},
  year={2022},
  organization={Springer}
}

@inproceedings{JLS18,
  author    = {Zhengfeng Ji and
               Yi{-}Kai Liu and
               Fang Song},
  editor    = {Hovav Shacham and
               Alexandra Boldyreva},
  title     = {Pseudorandom Quantum States},
  booktitle = {Advances in Cryptology - {CRYPTO} 2018 - 38th Annual International
               Cryptology Conference, Santa Barbara, CA, USA, August 19-23, 2018,
               Proceedings, Part {III}},
  series    = {Lecture Notes in Computer Science},
  volume    = {10993},
  pages     = {126--152},
  publisher = {Springer},
  year      = {2018},
  doi       = {10.1007/978-3-319-96878-0_5},
  bibsource = {dblp computer science bibliography, https://dblp.org}
}

@book{nielsen_chuang_2010, place={Cambridge}, title={Quantum Computation and Quantum Information: 10th Anniversary Edition}, DOI={10.1017/CBO9780511976667}, publisher={Cambridge University Press}, author={Nielsen, Michael A. and Chuang, Isaac L.}, year={2010}}

@inproceedings{BS19,
  author    = {Zvika Brakerski and
               Omri Shmueli},
  editor    = {Dennis Hofheinz and
               Alon Rosen},
  title     = {(Pseudo) Random Quantum States with Binary Phase},
  booktitle = {Theory of Cryptography - 17th International Conference, {TCC} 2019,
               Nuremberg, Germany, December 1-5, 2019, Proceedings, Part {I}},
  series    = {Lecture Notes in Computer Science},
  volume    = {11891},
  pages     = {229--250},
  publisher = {Springer},
  year      = {2019},
  doi       = {10.1007/978-3-030-36030-6_10},
  bibsource = {dblp computer science bibliography, https://dblp.org}
}

@inproceedings{AQY21,
      title={Cryptography from Pseudorandom Quantum States.}, 
      author={Ananth, Prabhanjan and Qian, Luowen and Yuen, Henry},
      booktitle={CRYPTO},
      year={2022}
}

@inproceedings{GLSV21,
  author    = {Alex B. Grilo and
               Huijia Lin and
               Fang Song and
               Vinod Vaikuntanathan},
  editor    = {Anne Canteaut and
               Fran{\c{c}}ois{-}Xavier Standaert},
  title     = {Oblivious Transfer Is in MiniQCrypt},
  booktitle = {Advances in Cryptology - {EUROCRYPT} 2021 - 40th Annual International
               Conference on the Theory and Applications of Cryptographic Techniques,
               Zagreb, Croatia, October 17-21, 2021, Proceedings, Part {II}},
  series    = {Lecture Notes in Computer Science},
  volume    = {12697},
  pages     = {531--561},
  publisher = {Springer},
  year      = {2021},
  doi       = {10.1007/978-3-030-77886-6_18},
  bibsource = {dblp computer science bibliography, https://dblp.org}
}

@inproceedings{BartusekCKM21a,
  author    = {James Bartusek and
               Andrea Coladangelo and
               Dakshita Khurana and
               Fermi Ma},
  editor    = {Tal Malkin and
               Chris Peikert},
  title     = {One-Way Functions Imply Secure Computation in a Quantum World},
  booktitle = {Advances in Cryptology - {CRYPTO} 2021 - 41st Annual International
               Cryptology Conference, {CRYPTO} 2021, Virtual Event, August 16-20,
               2021, Proceedings, Part {I}},
  series    = {Lecture Notes in Computer Science},
  volume    = {12825},
  pages     = {467--496},
  publisher = {Springer},
  year      = {2021},
  doi       = {10.1007/978-3-030-84242-0_17},
  bibsource = {dblp computer science bibliography, https://dblp.org}
}

@article{NS03,
  title = {Bit-commitment-based quantum coin flipping},
  author = {Nayak, Ashwin and Shor, Peter},
  journal = {Phys. Rev. A},
  volume = {67},
  issue = {1},
  pages = {012304},
  numpages = {11},
  year = {2003},
  month = {Jan},
  publisher = {American Physical Society},
  doi = {10.1103/PhysRevA.67.012304},
  url = {https://link.aps.org/doi/10.1103/PhysRevA.67.012304}
}

@misc{Col23,
      author = {Andrea Coladangelo},
      title = {Quantum trapdoor functions from classical one-way functions},
      howpublished = {Cryptology ePrint Archive, Paper 2023/282},
      year = {2023},
      note = {\url{https://eprint.iacr.org/2023/282}},
      url = {https://eprint.iacr.org/2023/282}
}

@article{chen2024power,
  title={The power of a single Haar random state: constructing and separating quantum pseudorandomness},
  author={Chen, Boyang and Coladangelo, Andrea and Sattath, Or},
  journal={arXiv preprint arXiv:2404.03295},
  year={2024}
}

@article{giurgica2023pseudorandomness,
  title={Pseudorandomness from subset states},
  author={Giurgica-Tiron, Tudor and Bouland, Adam},
  journal={arXiv preprint arXiv:2312.09206},
  year={2023}
}

@article{GS19,
  title={Two-round multiparty secure computation from minimal assumptions},
  author={Garg, Sanjam and Srinivasan, Akshayaram},
  journal={Journal of the ACM},
  volume={69},
  number={5},
  pages={1--30},
  year={2022},
  publisher={ACM New York, NY}
}

@inproceedings{CLM23,
  title={Black-box separations for non-interactive classical commitments in a quantum world},
  author={Chung, Kai-Min and Lin, Yao-Ting and Mahmoody, Mohammad},
  booktitle={Annual International Conference on the Theory and Applications of Cryptographic Techniques},
  pages={144--172},
  year={2023},
  organization={Springer}
}

@inproceedings{BL19,
  title={k-round multiparty computation from k-round oblivious transfer via garbled interactive circuits},
  author={Benhamouda, Fabrice and Lin, Huijia},
  booktitle={Advances in Cryptology--EUROCRYPT 2018: 37th Annual International Conference on the Theory and Applications of Cryptographic Techniques, Tel Aviv, Israel, April 29-May 3, 2018 Proceedings, Part II 37},
  pages={500--532},
  year={2018},
  organization={Springer}
}

@inproceedings{CLOS02,
  title={Universally composable two-party and multi-party secure computation},
  author={Canetti, Ran and Lindell, Yehuda and Ostrovsky, Rafail and Sahai, Amit},
  booktitle={Proceedings of the thiry-fourth annual ACM symposium on Theory of computing},
  pages={494--503},
  year={2002}
}

@inproceedings{CF01,
  title={Universally composable commitments},
  author={Canetti, Ran and Fischlin, Marc},
  booktitle={Advances in Cryptology—CRYPTO 2001: 21st Annual International Cryptology Conference, Santa Barbara, California, USA, August 19--23, 2001 Proceedings 21},
  pages={19--40},
  year={2001},
  organization={Springer}
}

@article{JGMW23,
  title={Subset States and Pseudorandom States},
  author={Jeronimo, Fernando Granha and Magrafta, Nir and Wu, Pei},
  journal={arXiv preprint arXiv:2312.15285},
  year={2023}
}

@incollection{BFM19,
  title={Non-interactive zero-knowledge and its applications},
  author={Blum, Manuel and Feldman, Paul and Micali, Silvio},
  booktitle={Providing Sound Foundations for Cryptography: On the Work of Shafi Goldwasser and Silvio Micali},
  pages={329--349},
  year={2019}
}

@article{CM21,
  title={Quantum merkle trees},
  author={Chen, Lijie and Movassagh, Ramis},
  journal={arXiv preprint arXiv:2112.14317},
  year={2021}
}

@article{BFV19,
  title={Computational pseudorandomness, the wormhole growth paradox, and constraints on the AdS/CFT duality},
  author={Bouland, Adam and Fefferman, Bill and Vazirani, Umesh},
  journal={arXiv preprint arXiv:1910.14646},
  year={2019}
}

@misc{MNY23,
      title={Unconditionally Secure Commitments with Quantum Auxiliary Inputs}, 
      author={Tomoyuki Morimae and Barak Nehoran and Takashi Yamakawa},
      year={2023},
      eprint={2311.18566},
      archivePrefix={arXiv},
      primaryClass={quant-ph}
}

@article{Harrow13church,
  title={The church of the symmetric subspace},
  author={Harrow, Aram W},
  journal={arXiv preprint arXiv:1308.6595},
  year={2013}
}

\appendix
\section{Alternative proof of~\Cref{lem:Z_haar_indis}} \label{app:simpleproof}

\begin{proof}[Proof sketch of~\Cref{lem:Z_haar_indis}]
The first part of the proof is the same as in~\cite{Col23}. Here we introduce the required notations and omit the details. Let $d := 2^n$ and
\begin{align*}
\sigma 
& := \sum_{x\in\bit^n} \rho_x
= \sum_{x\in\bit^n} \Ex_{\ket{\psi}\gets \Haar(2^n)} \left[ (Z^x \otimes I^{\otimes m}) \ketbra{\psi}{\psi}^{\otimes m+1} (Z^x \otimes I^{\otimes m}) \right] \\
& = \Ex_{\Vec{t}\in\cI_{d,m+1}} \sum_{x\in\bit^n} \left[ (Z^x \otimes I^{\otimes m}) \ketbra{s(\Vec{t})}{s(\Vec{t})} (Z^x \otimes I^{\otimes m}) \right] \\
& = \frac{d}{\binom{d + m}{m+1}} \cdot \sum_{\Vec{t}\in\cI_{d,m+1}} \sum_{j\in\bit^n} (\ketbra{j}{j} \otimes I^{\otimes m}) \ketbra{s(\Vec{t})}{s(\Vec{t})} (\ketbra{j}{j} \otimes I^{\otimes m}) \\
& = \frac{d}{\binom{d + m}{m+1}} \cdot \sum_{j\in\bit^n} \sum_{0\le r\le m} \sum_{\Vec{t}\in T_{j,r}^m} \frac{r+1}{m+1} \ketbra{j}{j} \otimes \ketbra{s(\Vec{t})}{s(\Vec{t})}.
\end{align*}
So we have
\begin{align*}
\sigma^{-1/2}
= \sqrt{\frac{\binom{d + m}{m+1}}{d}} \cdot \sum_{j\in\bit^n} \sum_{0\le r\le m} \sum_{\Vec{t}\in T_{j,r}^m} \sqrt{\frac{m+1}{r+1}} \ketbra{j}{j} \otimes \ketbra{s(\Vec{t})}{s(\Vec{t})}.
\end{align*}
Note that $\sigma^{-1/2}$ is PSD with the largest eigenvalue $\norm{\sigma^{-1/2}} = \sqrt{\binom{d + m}{m+1}(m+1)/d}$ (when $r = 0$). In~\cite{Col23}, the main technicality is to show Equation~(28):
\[
\Ex_{x\gets\bit^n} \Tr( \rho_x\sigma^{-1/2}\rho_x\sigma^{-1/2} )
\leq C' \cdot \left( \frac{m}{d} + \frac{m^7}{d^3} \right),
\]
where $C' > 0$ is some constant. Here, we provide an alternative and simpler proof. Since $\sigma^{-1/2}$ and $\rho_x$ are both PSD, the matrix $\sigma^{-1/2}\rho_x\sigma^{-1/2}$ is PSD as well. As $\rho_x$ is a density matrix, we have
\[
\Tr( \rho_x \cdot \sigma^{-1/2}\rho_x\sigma^{-1/2} )
\leq \norm{ \sigma^{-1/2}\rho_x\sigma^{-1/2} }.
\]
Then we use the submultiplicativity of the operator norm to obtain
\begin{align*}
& \norm{ \sigma^{-1/2}\rho_x\sigma^{-1/2} } \\
& \leq \norm{\sigma^{-1/2}} 
\cdot \norm{Z^x \otimes I^{\otimes m}} 
\cdot \norm{\Ex_{\Vec{t}\in\cI_{d,m+1}}\ketbra{s(\Vec{t})}{s(\Vec{t})}} 
\cdot \norm{Z^x \otimes I^{\otimes m}} 
\cdot \norm{\sigma^{-1/2}} \\
& = \norm{\sigma^{-1/2}}^2 \cdot \norm{\Ex_{\Vec{t}\in\cI_{d,m+1}}\ketbra{s(\Vec{t})}{s(\Vec{t})}} \tag{\text{unitaries have a unit operator norm}} \\
& = \frac{\binom{d + m}{m+1}\cdot (m+1)}{d} \cdot \frac{1}{\binom{d + m}{m+1}}
= \frac{m+1}{d}.
\end{align*}
Hence, it holds that
\[
\Ex_{x\gets\bit^n} \Tr( \rho_x\sigma^{-1/2}\rho_x\sigma^{-1/2} )
\leq \frac{m+1}{d}. \qedhere
\]
\end{proof}

\end{document}